%% file: arxiv.tex
\documentclass{mathincs}

\pdfoutput=1

\usepackage[utf8]{inputenc}
\usepackage{stix}
\usepackage[scaled=0.86]{helvet} 

\usepackage[utf8]{inputenc}
\usepackage{stix}
\usepackage[scaled=0.86]{helvet} 
\usepackage[T1]{fontenc}
\usepackage{pbox}

\usepackage[scaled=0.85]{beramono}

\usepackage{amsmath,amssymb,amsthm}
\usepackage{xspace}

\usepackage{color}
\definecolor{gray}{gray}{0.4}

\usepackage{amscd}
\usepackage{enumerate}
\usepackage[noend]{algorithmic}
\usepackage{algorithm}

\usepackage[shortlabels]{enumitem} 

\usepackage{fixme}
\fxsetup{author={Thibaut},envlayout={color},%
  layout={footnote,marginclue},theme=color}

\FXRegisterAuthor{mf}{amf}{Maria}

\newcommand{\true}{\textsf{TRUE}\xspace}

\usepackage[mathscr]{euscript}

\usepackage{array}

\usepackage{mathtools}
\mathtoolsset{showonlyrefs,showmanualtags}

\newtheorem{theorem}{Theorem}[section]
\newtheorem{lemma}[theorem]{Lemma}
\newtheorem{proposition}[theorem]{Proposition}

\theoremstyle{definition}
\newtheorem{definition}[theorem]{Definition}
\theoremstyle{remark}
\newtheorem*{example}{Example}  
\newtheorem{remark}[theorem]{Remark}

\numberwithin{equation}{section}

\newcommand{\LM}{\mathrm{LM}}
\newcommand{\LT}{\mathrm{LT}}

\newcommand{\ZZ}{\mathbb{Z}}
\newcommand{\KK}{\mathbb{K}}
\newcommand{\NN}{\mathbb{N}}

\newcommand{\divides}{\mid}
\newcommand{\LC}{\mathrm{LC}}

\newcommand{\lcm}{\mathrm{lcm}}

\newcommand{\bfu}{\mathbf{u}}
\newcommand{\bfe}{\mathbf{e}}
\newcommand{\bfp}{\mathbf{p}}
\newcommand{\bfT}{\mathbf{T}}
\newcommand{\bfr}{\mathbf{r}}
\newcommand{\bfg}{\mathbf{g}}
\newcommand{\bfh}{\mathbf{h}}
\newcommand{\bfalpha}{\boldsymbol{\alpha}}
\newcommand{\bfbeta}{\boldsymbol{\beta}}

\newcommand{\bfv}{\mathbf{v}}
\newcommand{\bfL}{\mathbf{L}}
\newcommand{\bfq}{\mathbf{q}}

\newcommand{\bfz}{\mathbf{z}}

\newcommand{\modpol}[1]{\overline{#1}}
\newcommand{\bbfp}{\modpol{\bfp}}
\newcommand{\bbfq}{\modpol{\bfq}}
\newcommand{\bbfr}{\modpol{\bfr}}
\newcommand{\bbfalpha}{\modpol{\bfalpha}}
\newcommand{\bbfbeta}{\modpol{\bfbeta}}
\newcommand{\bbfu}{\modpol{\bfu}}
\newcommand{\bbfL}{\modpol{\bfL}}
\newcommand{\bbfe}{\modpol{\bfe}}

\newcommand{\bbfz}{\modpol{\bfz}}

\newcommand{\sig}{\mathfrak{s}}
\newcommand{\Sig}{S}

\newcommand{\RR}{R}
\renewcommand{\AA}{A}

\newcommand{\alname}[1]{\textsf{#1}\xspace}
\newcommand{\algocoefs}{\alname{LinDecomp}}

\newcommand{\algosatideal}{\alname{SatIdeal}}

\newcommand{\algosing}{\alname{1-SingularReducible}}
\newcommand{\algoreduce}{\alname{Reduce}}
\newcommand{\algosigreduce}{\alname{RegularReduce}}
\newcommand{\algosigmoller}{\alname{SigMöller}}

\newcommand{\sGB}{$\sig$-GB\xspace}
\newcommand{\sssGB}{s-s $\sig$-GB\xspace}

\newcommand{\SPol}{\textup{S-Pol}}

\newcommand{\wrt}{w.r.t.\xspace}
\newcommand{\resp}{resp.\xspace}
\newcommand{\eg}{\emph{e.g.}}
\newcommand{\ie}{\emph{i.e.}\xspace}

\newcommand{\Sat}{\mathrm{Sat}}

\newcommand{\Mon}{\mathrm{Mon}}

\usepackage{hyperref}

\hypersetup{bookmarks=true,colorlinks=true}

\newcommand{\smallvspace}{\vspace{0pt}}
\newcommand{\proofvspace}{\vspace{0pt}} 
\newcommand{\largevspace}{\vspace{-12pt}}

\newif{\ifappendix}
\appendixtrue

\begin{document}

\title[A Signature-based Algorithm for Computing Gröbner Bases over Principal Ideal Domains]{A Signature-based Algorithm for Computing \\ Gröbner Bases over Principal Ideal Domains}

\author{Maria Francis}
\address{Dept of Computer Science and Engineering, Indian Institute of Technology, Hyderabad, India\\
  \url{mariaf@iith.ac.in}}

\thanks{This work was started when the first author was supported by the Austrian FWF grant Y464.
  The second author is supported by the Austrian FWF grant F5004.}

\author{Thibaut Verron}
\address{Institute for Algebra, Johannes Kepler University, Linz, Austria\\
  \url{thibaut.verron@jku.at}}




\maketitle

\largevspace
\largevspace
\begin{abstract}
  Signature-based algorithms have become a standard approach for Gröbner basis computations for polynomial systems over fields, but how to extend these techniques to coefficients in general rings is not yet as well understood.

  In this paper, we present a proof-of-concept signature-based algorithm for computing Gröbner bases over commutative integral domains.
  It is adapted from a general version of Möller's algorithm (1988) which considers reductions by multiple polynomials at each step.
  This algorithm performs reductions with non-decreasing signatures, and in particular, signature drops do not occur.
  When the coefficients are from a principal ideal domain (\eg{} the ring of integers or the ring of univariate polynomials over a field), we prove correctness and termination of the algorithm, and we show how to use signature properties to implement classic signature-based criteria to eliminate some redundant reductions.
  In particular, if the input is a regular sequence, the algorithm operates without any reduction to 0.

  We have written a toy implementation of the algorithm in Magma.
  Early experimental results suggest that the algorithm might even be correct and terminate in a more general setting, for polynomials over a unique factorization domain (\eg{} the ring of multivariate polynomials over a field or a PID).
\end{abstract}

\largevspace
\section{Introduction}

\newlength{\oldparskip}
\setlength{\oldparskip}{\parskip}
\setlength{\parskip}{1pt}

The theory of Gröbner bases was introduced by Buchberger in 1965 \cite{Buchberger:1965:thesis} and has since become a fundamental algorithmic tool in computer algebra.
Over the past decades, many algorithms have been developed to compute Gröbner bases more and more efficiently.
The latest iteration of such algorithms is the class of signature-based algorithms, which introduce the notion of signatures and use it to detect and prevent unnecessary or redundant reductions.
Following early work in \cite{MollerMoraTraverso1992}, the technique of signatures was first formally introduced for Algorithm F5~\cite{Faugere:2002:F5}, allowing to compute a Gröbner basis for a regular sequence without any reduction to zero.
Since then, there have been many research works in this direction~\cite{Gao:2010:G2V,arri:2011:f5,eder:2010:f5c,eder:2011:signature}.

All these algorithms are for ideals in polynomial rings over fields.
Gröbner bases can be defined and computed over commutative rings~\cite[Ch.~4]{Adams:1994:introtogrobnerbasis}.
This can be used in many applications, \eg{} for polynomials over $\ZZ$ in lattice-based cryptography~\cite{FrancisDukkipati:2014:Hash} or for polynomials over a polynomial ring as an elimination tool~\cite{Nabeshima:2009:polyringoverpolyring}.
Many other examples are described in~\cite{Lichtblau:applications}.

If the coefficient ring is not a field, there are two ways to define Gröbner bases, namely weak and strong bases.
Strong Gröbner bases ensure that normal forms can be computed as in the case of fields.
But a strong Gröbner basis is in general larger than a weak one, and if the base ring is not a Principal Ideal Domain (PID), then some ideals exist which do not admit a strong Gröbner basis.
On the other hand, weak Gröbner bases, or simply Gröbner bases, always exist for polynomial ideals over a Noetherian commutative ring.
They do not necessarily define a unique normal form, but they can be used to decide ideal membership.
If necessary, over a PID, a post-processing phase performing coefficient reductions can be used to obtain a strong Gröbner basis.

Recent works have focused on generalizing signature-based techniques to Gröbner basis algorithms over rings.
First steps in this direction, adding signatures to a modified version of Buchberger's algorithm for strong Gröbner bases over Euclidean rings~\cite{Lichtblau}, were presented in \cite{Eder:2017:EuclideanRings}.
That paper proves that a signature-based Buchberger's algorithm for strong Gröbner bases cannot ensure correctness of the result after encountering a ``signature-drop'', but can nonetheless be used as a prereduction step in order to significantly speed up the computations.

In this paper, we prove that it is possible to compute a weak signature-Gröbner bases of polynomial ideals over PIDs (including Euclidean rings) using signature-based techniques.
The proof-of-concept algorithm that we present is adapted from the weak Gröbner basis algorithm due to Möller~\cite{Moller:1988:grobnerrings2} \cite[Sec.~4.2]{Adams:1994:introtogrobnerbasis}, which is designed to compute a basis for a polynomial ideal over any ring, and does so by considering combinations and reductions by multiple polynomials at once.
The main difference with the results of~\cite{Eder:2017:EuclideanRings} is that we use a stricter definition of regular reductions, effectively preventing more reductions from happening, and at the same time adding more polynomials to the basis.

This ensures that no reductions leading to signature-drops can happen in the algorithm,
and as a consequence, we prove that the algorithm terminates and computes a signature Gröbner basis with elements ordered with non-decreasing signatures.
This property allows us to examine classic signature-based criteria, such as the syzygy criterion, the F5 criterion and the singular criterion, and show how they can be adapted to the case of PIDs.
In particular, when the input forms a regular sequence, the algorithm performs no reductions to zero.
To the best of our knowledge, this is the first algorithm that, given a regular sequence of polynomials with coefficients in a PID, can compute a Gröbner basis of the corresponding ideal without any reduction to zero.

Möller also presented an efficient algorithm that computes (strong)
Gr\"obner basis for polynomial ideals where the coefficients are from
Principal Ideal Rings~\cite[Section 4]{Moller:1988:grobnerrings2}.
That algorithm skips the combinatorial bottleneck of computing
saturated sets.
Instead, it uses two polynomials to build S-polynomials and makes use
of Gebauer-Möller criteria~\cite{gebauer:1988:installation},
previously introduced for fields, to discard redundant S-polynomials.

Whenever necessary, for clarity, we shall refer to that algorithm as
Möller's \emph{strong} algorithm.
The algorithm at the center of our focus, computing weak Gröbner bases, will be referred to as Möller's \emph{weak} algorithm, or simply Möller's algorithm.

We have written a toy implementation of the algorithms presented, with the F5 and the singular criteria, in the Magma Computational Algebra System~\cite{Magma}, and compared its efficiency, in terms of number of excluded pairs, with Möller's strong algorithm.
Experimentally, on all considered examples, Möller's (weak) algorithm with signatures does compute and reduce fewer S-polynomials than Möller's strong algorithm.

Möller's (weak) algorithm, without signatures, works for polynomial systems over any Noetherian commutative ring.
The signature-based algorithm is only proved to be correct and to terminate for PIDs, but with very few changes, it can be made to accommodate inputs with coefficients in a more general ring.
Interestingly, early experimental data with coefficients in a multivariate polynomial ring (a Unique Factorization Domain but not a PID) suggest that the signature-based algorithm might work over more general rings than just PIDs.
For that reason, and because it does not over-complicate the exposition, we choose to present Möller's algorithms, with and without signatures, in their most general form, accepting input over any Noetherian commutative ring. 



\paragraph{Previous works}
\label{sec:Previous-works}

Signature-based Gröbner basis algorithms over fields have been extensively studied, and an excellent survey of those works can be found in \cite{eder:2017:survey}. The technical details of most proofs can be found in \cite{practicalgrobner:2012:stillman,eder:2013:signature}.
The theory of Gr\"obner bases for polynomials over Noetherian commutative rings dates back to the 1970s \cite{Zacharias:1978:grobnerbasisrings3,Moller:1988:grobnerrings2} and a good exposition of these approaches can be found in \cite{Adams:1994:introtogrobnerbasis}.
Algorithms exist for both flavors of Gröbner bases: Buchberger's algorithm~\cite{Buchberger:1965:thesis} computes weak Gröbner bases over a PID, and Möller's weak algorithm~\cite{Moller:1988:grobnerrings2} extends this approach to Noetherian commutative rings.
As for strong Gröbner bases, they can be computed using an adapted version of Buchberger's algorithm~\cite{Kandri-Rody-Kapur} or Möller's strong algorithm~\cite{Moller:1988:grobnerrings2}.
Algorithms for computing signature Gröbner bases over Euclidean rings have been investigated in~\cite{Eder:2017:EuclideanRings}.

%



\setlength{\parskip}{\oldparskip}

\section{Notations}
\label{sec:Notations}

Let $\RR$ be a Noetherian integral domain, which is assumed to have a unit and be commutative.
Let $\AA = \RR[x_{1},\dots,x_{n}]$ be the polynomial ring in $n$ indeterminates $x_1,\ldots,x_n$ over $\RR$.
A monomial in $\AA$ is $x^{a} = x_{1}^{a_{1}}\dots x_{n}^{a_{n}}$ where $a=(a_{1},\dots,a_{n})\in \NN^{n}$.
A term is $kx^{a}$, where $k \in \RR$ and $k \neq 0$. 
We will denote all the terms 
in $\AA$ by $\mathrm{Ter}(\AA)$ and all the monomials in 
$\AA$ by $\mathrm{Mon}(\AA)$.
We use the notation $\mathfrak{a}$ for polynomial ideals in $\AA = \RR[x_1,\ldots,x_n]$  and $I$ for ideals in the coefficient ring $\RR$.  

The notion of monomial order can be directly extended from $\KK[x_1, \ldots, x_n]$ to $\AA$.
In the rest of the paper, we assume that $\AA$ is endowed with an implicit monomial order $\prec$, and we define as usual the leading monomial $\LM$, the leading term $\LT$ and the leading coefficient $\LC$ of a given polynomial.

Given a tuple of polynomials $(g_{1},\dots,g_{s})$ and $i \in \{1,\dots,s\}$, we will frequently denote, for brevity, $M(i) = \LM(g_{i})$, $C(i) = \LC(g_{i})$ and $T(i) = \LT(g_{i}) = C(i) M(i)$.


\largevspace
\section{Gr\"obner Bases in Polynomial Rings over $\RR$}
\label{sec:Grobn-Bases-Polyn}

For more details about the contents of this section, one can refer to 
\cite[Chapter 4]{Adams:1994:introtogrobnerbasis}.

\subsection{Computations in $\RR$}
\label{sec:Computations-A}

We assume that our coefficient ring $\RR$ is \emph{effective} in the following sense.
\begin{enumerate}[(1)]
  \item There are algorithms for arithmetic operations ($+$, ${\ast}$, zero test) in $\RR$.
  \item There is an algorithm~\algocoefs:
  \begin{itemize}
    \item Input: $\{k_{1},\dots,k_{s}\} \subset \RR$, $k \in \RR$
    \item Output: \true iff $k \in \langle k_{1},\dots,k_{s} \rangle$ and if yes, $l_1, \ldots, l_s \in \RR$ such that $k = k_1l_1+ \cdots + k_sl_s$.
  \end{itemize}
  \item There is an algorithm~\algosatideal:
  \begin{itemize}
    \item Input:  $\{k_{1},\dots,k_{s}\} \subset \RR$, $k \in \RR$
    \item Output: $\{l_{1},\dots,l_{r}\} \subset \RR$ 
    generators of the saturated ideal
    $\langle k_{1},\dots,k_{s} \rangle : \langle k \rangle$.
  \end{itemize}
\end{enumerate}
The condition that an algorithm~\algocoefs exists is called \emph{linear equations being solvable in $\RR$} in~\cite[Def.~4.1.5]{Adams:1994:introtogrobnerbasis}.

\begin{example}
  %
  Euclidean rings are effective, because one can implement those algorithms using GCD computations and Euclidean reductions.
  For example over $\ZZ$,
  $\algocoefs(\{ 4 \}, 12)$ is $(\true,\{3\})$, since 12 is in the ideal $\langle 4 \rangle$ and $12 = 3 \cdot 4$.
  The output of $\algosatideal(\{ 4 \},6)$ is $\{2\}$ since $\langle 4 \rangle : \langle 6 \rangle = \frac{1}{6}(\langle 4 \rangle \cap \langle 6 \rangle) = \frac{1}{6}\langle 12 \rangle = \langle 2 \rangle $.

  The ring of multivariate polynomials over a field is also effective, using Gröbner bases and normal forms to perform the same ideal computations.
\end{example}

\subsection{Weak Gröbner bases over rings}
\label{sec:Polynomial-reduction}

For reduction in fields it is enough to check if the leading term of $f$ is divisible by the leading monomial of $g$ even though the actual reduction happens with the leading term of $g$.
Clearly, in rings this is not a sufficient condition : $\LC(g)$ may not divide $\LC(f)$ even if $\LM(g)$ divides $\LM(f)$.
Requiring that $\LT(g)$ divide $\LT(f)$ leads to the notion of strong Gröbner basis, more details can be found in~\cite[Sec.~4.5]{Adams:1994:introtogrobnerbasis}.

Here we are interested in computing weak Gröbner bases, and we recall the main definitions in this section.
First, following~\cite{Moller:1988:grobnerrings2,Adams:1994:introtogrobnerbasis}, we expand the definition of reduction to allow for a linear combination of reducers.
We define saturated sets~\cite[Def.4.2.4]{Adams:1994:introtogrobnerbasis} (called maximal sets in \cite{Moller:1988:grobnerrings2}).

\begin{definition}\label{saturation}
  Given a tuple of monomials $(x^{a_1}, \ldots, x^{a_s})$, the \emph{saturated set} for a monomial $x^b$ \wrt $(x^{a_{1}},\dots,x^{a_{s}})$ is defined as
  \begin{equation}
    \label{eq:2}
    \Sat(x^b ; x^{a_{1}},\dots,x^{a_{s}})= \{i \in \{1, \ldots, s\} : x^{a_i} \divides x^b\}.
  \end{equation}
  A set $J \subseteq \{1, \ldots, s\}$ is said to be \emph{saturated} w.r.t. $(x^{a_1}, \ldots, x^{a_s})$ if $J = \Sat(M(J) ; x^{a_{1}},\dots,x^{a_{s}})$
  where $M(J) = \lcm(x^{a_{i}} : i \in J)$.
  When clear from the context, we shall omit the list of monomials and write $J_{x^{b}} = \Sat(x^{b})$.
  %
\end{definition}

Given a tuple of polynomials $(f_{1},\dots,f_{s})$ and a set of indices $J \subset \{1,\dots,s\}$, we denote by $I_{J}$ the ideal of $\RR$ defined as
\(  I_{J} := \langle \LC(f_{i}) : i \in J \rangle \)
and we define $M(J) = \lcm(\LM(f_{1}),\dots,\LM(f_{s}))$.

\begin{definition}
  \label{def:reductionthatweneed}
  Let $f \in \AA$.
  Let $f_{1},\dots,f_{s} \in \AA$ and $x^{a_{1}},\dots,x^{a_{s}} \in \Mon(\AA)$ be such that $x^{a_{i}}\LM(f_{i}) = \LM(f)$ for all $i$.
  We say that we can \emph{weakly top reduce} 
  $f$ by $f_1, \ldots, f_s \in \AA$ if there exist $l_{1}, \dots,l_{s}$ in $\RR$ such that
  \begin{equation}
    \label{eq:1}
    \LT(f) = \sum_{i=1}^{s} l_i x^{a_i}\LT({f_i}).
  \end{equation}
  In our setting we will only perform top reductions, so we will simply call them \emph{weak reductions}.

  The outcome of the total reduction step is $g = f - \sum_{i=1}^{s} l_{i}x^{a_{i}}f_{i}$ and the $f_{i}$'s are called the \emph{weak reducers}.
  A polynomial $f \in \AA$ is \emph{weakly reducible} if it can be weakly reduced, otherwise it is \emph{weakly reduced}.
\end{definition}

If $g$ is the outcome of reducing $f$, then $\LM(g) \prec \LM(f)$.

\begin{example}
  \label{sec:Polynomial-reduction-1}
  Consider the polynomial ring $\ZZ[x,y]$ with the lex ordering $y \prec x$, 
  and consider the set $F=\{f_1,f_2,f_3,f_4,f_5\}$ in $\mathbb{Z}[x,y]$, with $f_1 = 4xy +x, f_2 = 3x^2+y,f_3=5x,f_4=4y^2+y,f_5=5y$.
  Let $f = 2xy + 13y -5$. We have $\LT(f) = 2xy = (2y)  \LT (f_3) - (2) \LT(f_1)$. This implies we can weakly reduce $f$ with $f_1, f_3$ to get $g = f - (2y  f_3 - 2 f_1) = 2x+13y-5$.
\end{example}

We are now prepared to give the definition of (weak) Gr\"obner bases for an ideal in $\AA$.
\begin{definition}
  Let $\mathfrak{a}$ be an ideal in $\AA$ and $G=\{g_1,\ldots, g_t\}$ be a finite set of nonzero polynomials in $\mathfrak{a} $. 
  The set $G$ is called a \emph{weak Gr\"obner basis} of $\mathfrak{a}$ in $\AA$ if it satisfies the following equivalent properties.
  \begin{enumerate}
    \item $\langle\LT(G) \rangle= \langle \LT(\mathfrak{a}) \rangle$;
    \item for any $f \in \mathfrak{a}$, $f$ is weakly reducible modulo $G$;
    \item for any $f \in \AA$, $f \in \mathfrak{a}$ if and only if $f$ weakly reduces to $0$ modulo $G$.
  \end{enumerate}
  
\end{definition}




\begin{remark}
  Even though the notion of weak Gröbner bases is a weaker notion than that of strong Gröbner bases, one can use weak polynomial reductions to test for ideal membership.
  One can also define normal forms modulo a polynomial ideal.
  However, for those normal forms to be unique, one needs to perform further reductions on the coefficients, to ``coset representative form'', and one needs to perform reductions on non-leading coefficients as well~\cite[Th.~4.3.3]{Adams:1994:introtogrobnerbasis}.
  Finally, note that, over a PID, one can easily recover a strong basis from a weak one~\cite[Th.~4]{Moller:1988:grobnerrings2}.
\end{remark}

\subsection{Möller's algorithm for general rings}
\label{sec:Moll-algor-gener}

In this section, we present Möller's (weak) algorithm~\cite{Moller:1988:grobnerrings2} for computing Gröbner bases over rings satisfying the conditions of Sec.~\ref{sec:Computations-A}.
This algorithm is analogous to Buchberger's algorithm for rings, where the polynomial reduction is as defined above and S-polynomials are replaced with linear combinations of several (possibly more than $2$) polynomials, 
defined in the following sense.

Consider a set $\{g_{1},\dots,g_{t}\}$ of polynomials.
For $i \in \{1,\dots,t\}$, let $M(i) = \LM(g_{i})$, $C(i) = \LC(g_{i})$ and $T(i) = \LT(g_{i})$.
Let $J$ be a saturated subset of $\{1,\dots,t\}$ \wrt $\{M(1),\dots,M(t)\}$.
Recall that $M(J) = \lcm(M(j) : j \in J)$.
By definition, for all $j \in J$, $M(j)$ divides $M(J)$ and $J$ is maximal with this property.

Let $s \in J$ and $J^{\ast} = J \setminus \{s\}$.
Similar to the idea behind S-polynomials,
we want to eliminate the leading term  $C(s) M(J)$ of $\frac{M(J)}{M(s)}g_{s}$.
This can only be done if we multiply $\frac{M(J)}{M(s)}g_{s}$ by an element of the saturated ideal $\langle C(i) : i \in J, i \neq s \rangle : \langle C(s) \rangle$.
We want to consider all such multipliers, so we need to consider generators of this saturated ideal.

Let $c$ be such a generator, by definition $cC(s) \in \langle C(i) : i \in J, i \neq s \rangle$
so there exists $(b_{i})_{i \in J^{\ast}} \in \RR$ such that
\(c C(s) = \sum_{i \in J^{\ast}} b_{i} C(i)\).
The (weak) S-polynomial associated with $J$, $s$ and $c$, for some suitable $(b_{i})$, is defined as
\begin{equation}
  \label{eq:6}
  \SPol((g_{i})_{i \in J^{\ast}};g_{s};c
  ) = c \frac{M(J)}{M(s)}g_{s} - \sum_{i \in J^{\ast}} b_{i} \frac{M(J)}{M(i)}g_{i}.
\end{equation}

If the ring $\RR$ is a PID, the saturated ideal $\langle C(i) : i \in J, i \neq s \rangle : \langle C(s) \rangle$ admits a unique generator $c$ 
and we define
\begin{gather}
  \label{eq:75}
  C(J;s) = \LC(c g_{s}) = c C(s) = \lcm(\gcd(\{C(j) : j \in J^{\ast}\}),C(s)) \\
  \label{eq:76}
  T(J;s) = \LT(c g_{s}) = C(J;s) M(J).
\end{gather}
Then the S-polynomial associated with $J$, $s$, $c$, for some suitable $(b_{i})$, can be written in the following form
\begin{equation}
  \label{eq:77}
  \SPol((g_{i})_{i \in J^{\ast}};g_{s}) = \frac{T(J;s)}{T(s)} g_{s} - \sum_{i \in J^{\ast}} b_{i} \frac{M(J)}{M(i)} g_{i}.
\end{equation}


Using this definition of S-polynomials, we recall Möller's algorithm (Algo.~\ref{algo:Moller}) for computing a Gröbner basis of an ideal given by a set of generators over $\RR$.
The correctness and termination of this algorithm are shown in ~\cite[Th.~4.2.8 and Th.~4.2.9]{Adams:1994:introtogrobnerbasis}.

\begin{algorithm}
  \caption{Möller's  algorithm~\cite[Algo.~4.2.2]{Adams:1994:introtogrobnerbasis},
    \cite{Moller:1988:grobnerrings2}}
  \begin{algorithmic}\label{algo:Moller}
    \STATE \textbf{Input} 
    $F=\{f_1,\ldots,f_m\} \subseteq \AA \setminus \{0\}$, $\prec$ a monomial order on $\AA$ 
    \STATE \textbf{Output} $G = \{g_1, \ldots, g_t\}$, a Gr\"obner basis of $\langle F \rangle$\\[5pt]
    \STATE $G \leftarrow F$, $\sigma \leftarrow 1$, $s \leftarrow m$
    \WHILE {$ \sigma \neq s $}
    \STATE $\mathcal{S} \leftarrow
    \big\{ \text{subsets of } \{1, \ldots, \sigma\} \text{ saturated }
    \text{\wrt{} } \LM(g_1), \ldots,
    \LM(g_\sigma)
    \text{ which contain } \sigma \big\}
    $
    
    \FOR {\textbf{each} $J \in \mathcal{S}$}
    \STATE $M(J) \leftarrow \mathrm{lcm}(\LM(g_j): j \in J)$
    \STATE $J^{\ast} \leftarrow J \setminus \{\sigma\}$
    \STATE $\{c_{1},\dots,c_{\mu}\} \leftarrow \algosatideal(\{\LC(g_{j}) : 
    {j \in J^{\ast}}\}, \LC(g_{\sigma}))$
    \STATE \COMMENT {$\langle c_{1},\dots,c_{\mu} \rangle = \langle C(j) : j \in J^{\ast} \rangle : \langle C(\sigma) \rangle$}
    
    \FOR[For PIDs, $\mu=1$]{$i \in \{1, \ldots, \mu\}$} 
    \STATE $p \leftarrow \SPol((g_{j})_{j \in J^{\ast}};g_{\sigma};c_{i})$
    \STATE $r \leftarrow \algoreduce(p,G)$
    \IF{$r \neq 0$} 
    \STATE $g_{s+1} \leftarrow r$, $G \leftarrow G \cup \{g_{s+1}\}$, $s \leftarrow s+1$
    \ENDIF
    \ENDFOR
    \ENDFOR
    \STATE $\sigma \leftarrow \sigma+1$
    \ENDWHILE
    \RETURN $G$
  \end{algorithmic}
\end{algorithm}

\begin{algorithm}
  \caption{\algoreduce (Def.~\ref{def:reductionthatweneed})}
  \label{algo:reduce}
  \begin{algorithmic}
    \STATE \textbf{Input} 
    $G=\{g_1,\ldots,g_s\} \subseteq \AA \setminus \{0\}$, $\prec$ a monomial order on $\AA$  
    \STATE \textbf{Output} $r$ result of reducing $p$ modulo $G$\\[5pt]
    \STATE $\textit{reducible} \leftarrow \true$, $r \leftarrow p$
    \WHILE {$\textit{reducible}$ \textbf{is} \true}
    \STATE $J \leftarrow \{j \in \{1,\dots,s\} : \LM(g_{j}) \divides \LM(r)\}$
    \STATE $\textit{reducible},(k_{j})_{j \in J} \leftarrow \algocoefs(\{\LC(g_{j}) : j \in J\}, \LC(r))$
    \STATE \COMMENT {If \textit{reducible} is \true, then $\LC(r) = \sum_{j \in J} k_{j}\LC(g_{j})$}
    \IF {$\textit{reducible}$}
    \STATE $r \leftarrow  r - \sum_{j \in J} k_{j} \frac{\LM(r)}{\LM(g_j}g_{j}$
    \ENDIF
    \ENDWHILE
    \RETURN $r$
  \end{algorithmic}
\end{algorithm}

\largevspace
\section{Signatures in $\AA^m$}
\label{sec:Sign-Based-Algor}

We consider the free $\AA$-module $\AA^{m}$ with basis $\bfe_1, \ldots, \bfe_m$.
A term (\resp{} monomial) in $\AA^{m}$ is $kx^{a}\bfe_{i}$ (\resp{} $x^{a}\bfe_{i}$) for some $k \in \RR \setminus \{0\}$, $x^{a} \in \Mon(\AA)$, $i \in \{1,\dots,m\}$.
In this paper, terms in $\AA^{m}$ are ordered using the Position Over Term (POT) order, defined by
\begin{displaymath}
  k x^a\bfe_i \prec lx^b\bfe_j \iff i \lneq j (\text{ or } i = j \text{ and } x^a \prec x^b).
\end{displaymath}
Given two terms $kx^{a}\bfe_{i}$ and $lx^{b}\bfe_{j}$ in $\AA^{m}$, we write $kx^{a}\bfe_{i} \simeq lx^{b}\bfe_{j}$ if they are incomparable, \ie{} if $a=b$ and $i=j$.

Given a set of polynomials $f_{1},\dots,f_{m} \in \AA$, elements of
$A^{m}$ encode elements of the ideal $\langle f_{1},\dots,f_{m}
\rangle$ through the $\AA$-module homomorphism $\bar{\cdot} : \AA^{m}
\to \AA$, defined by setting $\bbfe_{i} = f_{i}$ and extending
linearly to $\AA^{m}$.
In particular, $\overline{\sum_{i=1}^{m} p_{i}\bfe_{i}} = \sum_{i=1}^{m} p_{i}f_{i}$.

We recall the concept of signatures in $\AA^m$.
Let $\bfp = \sum_{i=1}^{m} p_{i} \bfe_{i}$ be a module element.
Under the POT ordering, the signature of $\bfp$ is $\sig(\bfp) = \LT(p_{i})\bfe_{i}$ where $i$ is such that $p_{i+1}{=}\dots{=}p_{m}=0$ and $p_{i}\neq 0$.
Signatures are of the form $kx^a\bfe_i$, where $k \in \RR, x^a \in \mathrm{Mon}(\AA)$ and $\bfe_i$ is a standard basis vector.

Note that we have two ways of comparing two similar signatures $\sig(\bfalpha) = kx^a\bfe_i$ and $\sig(\bfbeta) = lx^b\bfe_j$.
We write $\sig(\bfalpha)=\sig(\bfbeta)$ if $k=l$, $a=b$ and $i=j$, and we write $\sig(\bfalpha) \simeq \sig(\bfbeta)$ if $a=b$ and $i=j$, $k$ and $l$ being possibly different.
If $\RR$ is a field, one can assume that the coefficient is $1$, and so this distinction is not important.

Note also that when we order signatures, we only compare the corresponding module monomials, and disregard the coefficients.
This is a different approach from the one used in \cite{Eder:2017:EuclideanRings}, where both signatures and coefficients are ordered.

Given a tuple $(\bfalpha_{1},\dots,\bfalpha_{s})$ of module elements in $\AA^{m}$ and $i,j \in \{1,\dots,s\}$, we shall frequently denote $S(i) = \sig(\bfalpha_{i})$ for brevity.

In order to keep track of signatures we modify Def.~\ref{def:reductionthatweneed} to introduce the notion of $\sig$-reduction.

\begin{definition}
  Let $\bfp \in \AA^{m}$.
  We say that we can \emph{signature-reduce} (or \emph{$\sig$-reduce})
  $\bfp$ by $\bfbeta_1, \ldots, \bfbeta_s\in \AA^m$ if
  we can reduce $\bbfp$ by $\bbfbeta_{1},\dots,\bbfbeta_{s}$
  (in the sense of Def.~\ref{def:reductionthatweneed}) and $\sig(x^{a_i}\bfbeta_i) \preceq \sig(\bfp)$ for 
  all $i = 1, \ldots, s$, where $x^{a_{i}}=\frac{\LM(\bbfp)}{\LM(\bbfbeta_{i})}$.
  We can define similarly $\sig$-reduced module elements.
  
  If $ \sig(x^{a_i}\bfbeta_i) \simeq \sig(\bfp)$ for some $i$ in the above $\sig$-reduction,
  then it is called a \emph{singular} $\sig$-reduction step.
  Otherwise it is called a \emph{regular} $\sig$-reduction step.
  
  If $\sig(x^{a_{i}}\bfbeta_{i}) \simeq \sig(\bfp)$ for exactly one $i$ and it is actually an equality
  $\sig(l_{i}x^{a_{i}}\bfbeta_{i}) = \sig(\bfp)$, it is called a \emph{1-singular} $\sig$-reduction step.
\end{definition}

\begin{remark}
  For simplicity, we only carry out weak top reductions, and in particular all $\sig$-reductions are weak top $\sig$-reductions.
  But performing regular $\sig$-reduction to eliminate trailing terms does not affect the correctness of the algorithm.
\end{remark}

Just like $\sig$-reduction over fields, one can interpret $\sig$-reduction as polynomial reduction with an extra condition on the signature of the reducers.
The difference with fields is that in $\RR[x_1, \dots,x_n]$ polynomial reduction is defined differently from the classic polynomial reduction.
Additionally, in the case of fields, all singular $\sig$-reductions are $1$-singular.

The outcome $\bfq$ of $\sig$-reducing $\bfp$ is such that $\LT(\bbfq) \prec \LT(\bbfp)$ and $\sig(\bfq) \preceq \sig(\bfp)$.
If $\bfq$ is the result of a regular $\sig$-reduction, then $\sig(\bfq) = \sig(\bfp)$.
In signature-based algorithms, in order to keep track of the signatures of the basis elements, we only allow regular $\sig$-reductions.
Later, we will also prove that elements which are $1$-singular $\sig$-reducible can be discarded.
\begin{remark}
  In \cite[Ex.~2]{Eder:2017:EuclideanRings}, a signature drop appears when $\sig$-reducing an element of signature $6y\bfe_2$ with an element of signature $y\bfe_2$ causing the signature to ``drop'' to $5y\bfe_2$. 
  With our definition, since we only compare the module monomial part of the signatures, this is a (forbidden) singular $\sig$-reduction. 
\end{remark}
\begin{definition}
  Let $\mathfrak{a} = \langle f_1, \dots, f_m\rangle$ be an ideal in $\AA$. A finite subset $\mathcal{G}$ of $\AA^m$ is a (weak) \emph{signature Gr\"obner basis} (or \sGB for short) of $\mathfrak{a}$ if all $\bfu \in \AA^m$ $\sig$-reduce to zero mod $\mathcal{G}$.

  Given a signature $\mathbf{T}$, we say that $\mathcal{G}$ is a (partial) signature Gröbner basis up to $\textbf{T}$ if all $\bfu \in \AA^{m}$ with signature $\prec \bfT$ $\sig$-reduce to 0 mod~$\mathcal{G}$.
\end{definition}

Using this definition, we can give the following characterization of 1-singular reducibility, which allows for an easy algorithmic test.

\begin{lemma}[Characterization of $1$-singular $\sig$-reducibility]
  \label{lemme:correctness-1sing}
  Let $\mathcal{G}=\{\bfalpha_{1},\dots,\bfalpha_{s}\} \subset \AA^{m}$ and $\bfp \in \AA^{m}$ such that 
  $\mathcal{G}$ is a signature Gröbner basis up to signature $\sig(\bfp)$.
  Then $\bfp$ is 1-singular $\sig$-reducible if and only if there exist $j \in \{1,\dots,s\}$ and $k \in \RR$ and a monomial $x^{a}$ in $\AA$ such that
  $\LM(x^{a}\bbfalpha_{j}) = \LM(\bbfp)$ and $kx^{a}\sig(\bfalpha_{j})=\sig(\bfp)$.
\end{lemma}
\begin{proof}
  If $\bfp$ is 1-singular $\sig$-reducible, then such $j$, $k$ and $x^{a}$ exist by definition.
  Conversely, given such $j$, $k$ and $x^{a}$, if $kx^{a}\LT(\bbfalpha_{j}) = \LT(\bbfp)$, then $\bfp$ is 1-singular $\sig$-reducible.
  If not, then $\LM(\bbfp-kx^{a}\bbfalpha_{j}) = \LM(\bbfp)$.
  Furthermore, $\sig(\bfp - kx^{a}\bfalpha_{j}) \prec \sig(\bfp)$, so $\bfp - kx^{a}\bfalpha_{j}$ $\sig$-reduces to $0$.
  In particular, there exist $(\mu_{i})_{i \in \{1,\dots,s\}}$ terms in $\AA$ such that for all $i$ with $\mu_{i} \neq 0$, $\LM(\mu_{i}\bbfalpha_{i})=\LM(\bbfp-kx^{a}\bbfalpha_{j})$,
  $\LT(\bbfp - kx^{a}\bbfalpha_{j}) = \sum_{i=1}^{s} \mu_{i}\LT(\bbfalpha_{i})$
  and $\mu_{i}\sig(\bfalpha_{i}) \preceq \sig(\bfp - kx^{a}\bfalpha_{j}) \prec \sig(\bfp)$.
  So putting together the two $\sig$-reductions, we obtain that
  \begin{equation}
    \label{eq:11}
    \LT(\bbfp)
    = kx^{a}\LT(\bbfalpha_{j}) + \sum_{i=1}^{s} \mu_{i}\LT(\bbfalpha_{i})
  \end{equation}
  and this is a 1-singular $\sig$-reduction of $\bfp$.
\end{proof}

We now define (weak) \emph{semi-strong signature Gröbner bases}, which form a subclass of weak $\sig$-Gröbner bases.
In the case of rings, it is easier to compute them than to directly compute weak $\sig$-Gröbner bases.

\begin{definition}
  Let $\mathfrak{a} = \langle f_1, \dots, f_m\rangle$ be an ideal in $\AA$. A finite subset $\mathcal{G}$ of $\AA^m$ is a \emph{semi-strong signature Gr\"obner basis} (or \sssGB for short) of $\mathfrak{a}$ if, for all $\bfu \in \AA^m$,
  \begin{itemize}
    \item either $\bfu$ is (weakly) regular $\sig$-reducible modulo $\mathcal{G}$;
    \item or $\bfu$ is 1-singular $\sig$-reducible modulo $\mathcal{G}$;
    \item or $\bbfu = 0$.
  \end{itemize}
  Given a signature $\bfT$, \emph{semi-strong signature Gröbner bases up to $\bfT$} are defined similarly by only considering module elements with signature $\prec \bfT$.
\end{definition}

\begin{lemma}[{\cite[Lem.~4.6]{eder:2017:survey}}]
  Let $\mathfrak{a} = \langle f_{1},\dots,f_{m} \rangle$ be an ideal in $\AA$ and let $\mathcal{G} \subset \AA^{m}$.
  Then
  \begin{enumerate}
    \item If $\mathcal{G}$ is a \sssGB of $\mathfrak{a}$, then $\mathcal{G}$ is a \sGB of $\mathfrak{a}$.
    \item If $\mathcal{G}$ is a \sGB of $\mathfrak{a}$, then $\{\overline{\bfalpha}: \bfalpha \in \mathcal{G}\}$ is 
    a Gr\"obner basis of~$\mathfrak{a}$. 
  \end{enumerate}
\end{lemma}
\begin{proof}  The definition of a semi-strong Gröbner basis implies that all $\bfu \in \AA^{m}$ with $\bbfu \neq 0$ are $\sig$-reducible modulo $\mathcal{G}$, and so such $\sig$-reductions form a chain which can only terminate at $0$.

  The proof that a signature Gröbner basis is a Gröbner basis is classical~\cite[Lem.~4.1]{eder:2017:survey}.
\end{proof}

In order to compute signature Gröbner bases, similar to the case of fields, we will restrict the computations to regular S-polynomials.
For this purpose, we first introduce the signature of a set of indices, and regular sets.
\begin{definition}
  Let $\mathcal{G} = (\bfalpha_{1},\dots,\bfalpha_{t})$ be a tuple of module elements in $A^{m}$ and a set $J \subseteq \{1, \dots, t\} $.
  For $i\in \{1,\dots,t\}$, let $M(i) = \LM(\bbfalpha_i)$, and $S(i) = \sig(\bfalpha_i)$.
  The \emph{presignature} of $J$ is defined as 
  \begin{equation}
    S_{J} = \max_{s \in J} \left\{ \frac{M(J)} { M(s)} \Sig(s)  \right\}.
  \end{equation}
  
  We say that $J$ is a \emph{regular set} if there exists exactly one $s \in J$
  such that $S_{J} \simeq \frac{M(J)} {M(s)} \sig(\bfalpha_{s})$.
  The index $s$ is called the \emph{signature index} of $J$.
  We say that $J$ is a \emph{regular saturated set} if $J \setminus \{s\}$ contains all $j$ such that $M(j) \divides M(J)$ and $\frac{M(J)} {M(j)}\Sig(j) \prec S_{J}$.
\end{definition}

Note that given a regular set $J$, one can always compute a regular saturated set $J'$ containing $J$, by adding those indices $j$ such that $M(j) \divides M(J)$ and $\frac{M(J)} { M(j)}\Sig(j) \prec S_{J}$.

\begin{definition}
  Let $(\bfalpha_{1},\dots,\bfalpha_{t})$ be a tuple of module elements in $A^{m}$.
  For $i \in \{1,\dots,t\}$, let $M(i) = \LM(\bbfalpha_{i})$, $C(i) = \LC(\bbfalpha_{i})$ and $S(i) = \sig(\bfalpha_{i})$.
  Let $J \subset \{1,\dots,t\}$ be a regular saturated set with signature index $s$, and let $J^{\ast} = J \setminus \{s\}$.
  Let $c$ be an element of a family of generators of $\langle C(j) : j \in J^{\ast}\rangle : \langle C(s)\rangle$.
  %
  Let $(b_j)_{j \in J^{\ast}}$ be a tuple of elements of $\RR$ such that
  \(cC(s) = \sum_{j \in J^{\ast}} b_jC(j).
  \)
  Then the (weak) S-polynomial associated with $J$ and $c$ is defined as
  \begin{displaymath}
    \SPol((g_{j})_{j \in J};c) = c\frac{M(J)}{M(s)}\bfalpha_s - \sum_{j \in J^{\ast}} b_j\frac{M(J)}{M(j)}\bfalpha_j.
  \end{displaymath}
  Its signature is
  \begin{equation}
    \label{eq:86}
    S(J;c) = \sig(\SPol((g_{j})_{j \in J};c)) = c S_{J} = c \frac{M(J)}{M(s)} \Sig(s).
  \end{equation}
\end{definition}

\begin{remark}
  When dealing with regular saturated sets, unlike in Sec.~\ref{sec:Polynomial-reduction}, we do not need to specify which $s \in J$ is singled out when computing the S-polynomial: the only possible $s$ is the signature index of $J$.
\end{remark}

\begin{remark}
  If the coefficient ring is a PID, the ideal $\langle C(j) : j \in J^{\ast}\rangle : \langle C(s)\rangle$ is principal, and $c$ is uniquely determined up to an invertible factor.
  As such, it can be omitted, and in that case we shall simply write $\SPol(J)$ for the S-polynomial, and $S(J)$ for its signature.
  The signature can then be written as
  \(S(J) = \frac{C(J)}{C(s)} S_{J} = \frac{C(J)}{C(s)} \frac{M(J)}{M(s)} \Sig(s).
  \)
\end{remark}

\largevspace
\section{Adding signatures to Möller's weak algorithm}
\label{sec:Comp-weak-sign}

Recall that all $\sig$-reductions are weak top $\sig$-reductions.
In this section, all S-polynomials are weak S-polynomials.

\subsection{Algorithms}
\label{sec:Algorithms}

Algorithm~\algosigmoller (Algo.~\ref{Algorithmsig}) is a signature-based version of M\"oller's algorithm which, given an ideal $\mathfrak{a}$ in $\RR[x_{1},\dots,x_{n}]$ where $\RR$ is a PID, computes a signature Gr\"obner basis of $\mathfrak{a}$.

The algorithm proceeds by maintaining a list of regular saturated sets $\mathcal{P}$ and computing weak S-polynomials obtained from these saturated sets.
At each step, it selects the next regular saturated set $J \in \mathcal{P}$ such that $J$ has minimal presignature amongst elements of $\mathcal{P}$.
This ensures that the algorithm computes new elements for the signature Gröbner basis with nondecreasing signatures (Prop.~\ref{lemme:nondecreasing}).

The algorithm then regular $\sig$-reduces these S-polynomials \wrt{} the previous elements, and adds to the basis those which are not equal to 0 and are not 1-singular $\sig$-reducible.
Signature-based Gröbner basis algorithms over fields typically discard all new elements which are singular $\sig$-reducible, but this may be too restrictive for rings.
On the other hand, the proof of Lem.~\ref{supportinglemma} justifies that 1-singular $\sig$-reducible module elements can be safely discarded in the computations.
The correctness of the criterion for 1-singular $\sig$-reducibility (Algo.~\ref{algo:1sing}) was justified in Lem.~\ref{lemme:correctness-1sing}.
The correctness and termination of Algorithm~\algosigmoller are proved in Th.~\ref{regularreducers} and Th.~\ref{lemmatermination} respectively.

\begin{algorithm}
  \caption{Signature-based Möller's algorithm (\algosigmoller)}
  \begin{algorithmic}\label{Algorithmsig}
    \STATE \textbf{Input} $F=\{f_1,\ldots,f_m\} \subseteq \AA \setminus \{0\}$, $\prec$ a monomial order on $\AA$ 
    \STATE \textbf{Output} $\mathcal{G} = \{\bfalpha_{1},\dots,\bfalpha_{t}\}$ a semi-strong signature-Gröbner basis of $
    \langle F \rangle$\\[5pt]
    \STATE $\mathcal{G} \leftarrow \emptyset$, $\sigma \leftarrow 0$
    \FOR{$i \in \{1, \ldots, m\}$}
    \STATE $\bfe'_{i} \leftarrow \alname{RegularReduce}(\bfe_{i},\mathcal{G})$
    \IF {$\bbfe'_{i} \neq 0$}
    \STATE $\mathcal{G} = \mathcal{G} \cup \{\bfe'_i\}$, $s \leftarrow |\mathcal{G} |$ \hspace{1cm}\COMMENT{$\bfalpha_{s} = \bfe'_i$}
    \STATE $\mathcal{P} \leftarrow \{\text{Regular saturated sets of $\{1,\dots,s\}$ containing $s$}\}$
    \WHILE {$\mathcal{P} \neq \emptyset$}
    \STATE Pick and remove from $\mathcal{P}$ a regular saturated set with minimal presignature $S_{J}$
    \STATE $M(J) \leftarrow \mathrm{lcm}(\LM(\bbfalpha_j): j \in J)$
    \STATE $\tau \leftarrow \text{signature index of $J$}$
    \STATE $J^{\ast} \leftarrow J \setminus \{\tau\}$
    \STATE $\{c_{1},\dots,c_{\mu}\} \leftarrow \algosatideal(\{\LC(\bbfalpha_{j}) : j \in J^{\ast}\},
    \LC(\bbfalpha_{\tau}))$
    \FOR[For PIDs, $\mu=1$]{$i \in \{1, \ldots, \mu\}$}
    \STATE $p \leftarrow \SPol((g_{j})_{j \in J};c_{i})$
    \STATE $\bfr \leftarrow \alname{RegularReduce}(\bfp,\mathcal{G})$
    \IF {$\bfr \neq 0$ \textbf{and} \textbf{not} $\alname{1-SingularReducible}(\bfr,\mathcal{G})$}
    \STATE $\bfalpha_{s+1} \leftarrow \bfr$ \COMMENT{$\bfalpha_{s+1}$ has signature $S(J) = c_{i}S_{J}$}
    \STATE $\mathcal{G} \leftarrow \mathcal{G} \cup \{\bfalpha_{s+1}\}$
    \STATE $\mathcal{P} \leftarrow \mathcal{P} \cup \{\text{Regular saturated sets of $\{1,\dots,s+1\}$ containing $s+1$}\}$
    \STATE $s \leftarrow s+1$
    \ENDIF
    \ENDFOR
    \ENDWHILE
    \ENDIF
    \ENDFOR
    \RETURN $\mathcal{G}$
  \end{algorithmic}
\end{algorithm}

Due to space constraints, the subroutine \algosigreduce is not explicitly written.
It implements regular $\sig$-reduction of a module element $\bfp$ \wrt{} a set of module elements $\{\bfalpha_{1},\dots,\bfalpha_{s}\}$.
It is a straightforward transposition of \algoreduce (Algo.~\ref{algo:reduce}), with the additional condition that we only consider as reducers of $\bfr$ those $\bfalpha_{j}$ with
\(  
\LM(\bbfalpha_{j}) \divides \LM(\bbfr) \text{ and } \frac{\LM(\bbfr)}{\LM(\bbfalpha_{j})}\sig(\bfalpha_{j}) \prec \sig(\bfr).
\)

\begin{remark}
  Note that the algorithms, as presented, perform computations on module elements.
  However, for practical implementations, this represents a significant overhead.
  On the other hand, for any module element $\bfalpha$, we only need its polynomial value $\bbfalpha$ and its signature $\sig(\bfalpha)$.
  Hence the algorithm only needs to keep track of the signatures of elements, which is made possible by the restriction to regular S-polynomials and regular $\sig$-reductions.
\end{remark}

\begin{example}
  \ifappendix
  An example run of Algorithm~\ref{Algorithmsig} is provided in Appendix~\ref{sec:Example-run-Algor}.
  \else
  An example run of Algorithm~\ref{Algorithmsig} is available online\footnote{\url{https://github.com/ThibautVerron/SignatureMoller/ACA18_example.pdf}}.
  \fi
\end{example}

\subsection{Proof of correctness}
\label{sec:Proof-correctness}

In this section we prove the correctness of the algorithms presented in Sec.~\ref{sec:Algorithms}.
The first result states that Algorithm \algosigmoller computes elements of the signature Gröbner basis in nondecreasing order on their signatures.

\begin{proposition}
  \label{lemme:nondecreasing}
  Let $(\bfalpha_{1},\dots,\bfalpha_{t})$ be the value of $\mathcal{G}$ at any point in the course of Algorithm~\algosigmoller.
  Then $\sig(\bfalpha_{1}) \preceq \sig(\bfalpha_{2}) \preceq \dots \preceq \sig(\bfalpha_{t})$.
\end{proposition}
\begin{proof}
  Assume that this is not the case, and let $i$ be the smallest index such that $\sig(\bfalpha_{i}) \succ \sig(\bfalpha_{i+1})$.
  Let $J_{i}$ (\resp $J_{i+1}$) be the saturated set used to compute $\bfalpha_{i}$ (\resp $\bfalpha_{i+1}$).
  Note that $\sig(\bfalpha_{i}) \simeq \Sig(J_{i})$ and $\sig(\bfalpha_{i+1}) \simeq \Sig(J_{i+1})$.
  
  If $i \notin J_{i+1}$, then $J_{i+1}$ was already in the queue $\mathcal{P}$ when $J_{i}$ was selected, and so, by the selection criterion in the algorithm, $\Sig(J_{i}) \preceq \Sig(J_{i+1})$.

  If $i \in J_{i+1}$, then $\Sig(J_{i+1}) \succeq \frac{x^{J_{i+1}}}{\LM(\bbfalpha_{i})}\sig(\bfalpha_{i}) \succeq \sig(\bfalpha_{i})$.
\end{proof}


The following useful lemma gives consequences of the fact that two regular $\sig$-reduced elements share the same signature.

\begin{lemma}\label{initialtermslemma}
  Let $\mathcal{G}= (\bfalpha_{1},\dots, \bfalpha_{s})$ be a signature Gröbner basis up to signature $\bfL$.
  Let $\bfp, \bfq \in \AA^{m}$ such that $\sig(\bfp) = \sig(\bfq) = \bfL$, and $\bfp$ and $\bfq$ are regular $\sig$-reduced.
  Then $\LM(\bbfp) = \LM(\bbfq)$ and either $\LT(\bbfp) = \LT(\bbfq)$, or $\LC(\bbfp - \bbfq)$ lies in the ideal
  \begin{equation}
    \label{eq:9}
    C := \Big\langle \LC(\modpol{\bfalpha_{j}}) : \LM(\modpol{\bfalpha_{j}}) \divides m
     \text{ and }
     \frac{m}{\LM(\modpol{\bfalpha_{j}})}\sig(\bfalpha_{j}) \not\simeq \sig(\bfp) \Big\rangle.
  \end{equation}
\end{lemma}
\begin{proof}
  Let $\bfr = \bfp - \bfq$.
  Since $\sig(\bfp) = \sig(\bfq)$, we have $\sig(\bfr) \prec \sig(\bfp) = \bfL$, and so $\bfr$ $\sig$-reduces to $0$ modulo $\mathcal{G}$.
  Assume first that $\LM(\modpol{\bfp}) \neq \LM(\modpol{\bfq})$,
  then w.l.o.g. we may assume that $\LM(\modpol{\bfp}) \succ \LM(\modpol{\bfq})$,
  so $\LM(\modpol{\bfr}) = \LM(\modpol{\bfp})$.
  Since $\bfr$ is regular $\sig$-reducible, $\bfp$ is $\sig$-reducible.
  This is a contradiction with the assumption that $\bfp$ is $\sig$-reduced.

  So $\LM(\modpol{\bfp}) = \LM(\modpol{\bfq}) =: m$.
  If $\LT(\bbfp) \neq \LT(\bbfq)$, $C$ is the ideal of leading coefficients of polynomials which can eliminate $m$, 
  and since $\bfr$ is $\sig$-reducible, $\LC(\bbfp) - \LC(\bbfq) \in C$.
  %
\end{proof}

\begin{algorithm}
  \caption{Test of 1-singular $\sig$-reducibility modulo a partial $\sig$-GB (\algosing)}
  \label{algo:1sing}
  \begin{algorithmic}\label{Algorithmsingular}
    \STATE \textbf{Input} $\mathcal{G}=\{\bfalpha_{1},\dots,\bfalpha_{s}\} \subset \AA^{m}$ and $\bfp \in \AA^{m}$ such that $\bfp$ is regular $\sig$-reduced \wrt{} $\mathcal{G}$ and $\mathcal{G}$ is a signature Gröbner basis up to $\sig(\bfp)$
    \STATE \textbf{Output} \true iff $\bfp$ is 1-singular $\sig$-reducible modulo $\mathcal{G}$\\[5pt]
    \STATE $J \leftarrow \left\{j \in \{1,\dots,s\} : \LM(\bbfalpha_{j}) \divides \LM(\bbfp)
     \text{ and } \frac{\LM(\bbfp)}{\LM(\bbfalpha_{j})}\sig(\bfalpha_{j}) \preceq \sig(\bfp)\right\}$
    \RETURN $\exists j \in J, \exists k_{j} \in \RR, k_{j}\frac{\LM(\bbfp)}{\LM(\bbfalpha_{j})}\sig(\bfalpha_{j}) = \sig(\bfp)$
  \end{algorithmic}
\end{algorithm}

We now prove the correctness of Algorithm~\algosigmoller.
The proof follows the structure of the proof in the case of fields~\cite{practicalgrobner:2012:stillman}, and adapts it to Möller's algorithm over PIDs.
In particular, it takes into account weak $\sig$-reductions instead of classical $\sig$-reductions.
The algorithm ensures that all regular S-polynomials up to a given signature $\mathbf{T}$ $\sig$-reduce to 0, and proving the correctness of the algorithm requires proving that this implies that all module elements with signature $ \prec \mathbf{T}$ $\sig$-reduce to $0$.

The key lemma of the proof is the following.

\begin{lemma}\label{supportinglemma}
  Let $\mathcal{G}= (\bfalpha_{1},\dots, \bfalpha_{s}) \subseteq \AA^{m}$.
  Let $\bfu \in \AA^{m} \setminus \{0\}$ be $\sig$-reduced such that $\bbfu \neq 0$.
  Assume that $\mathcal{G}$ is a \sssGB{} basis up to signature $\sig(\bfu)$.
  Then there exists an S-polynomial $\bfp$ w.r.t. $\mathcal{G}$, such that:
  \begin{enumerate}
    \item the signature of $\bfp$ divides the signature of $\bfu$: $kx^a\sig(\bfp) = \sig(\bfu)$ with $k \in \RR$ and $x^{a} \in \Mon(\AA)$;
    \item if $\bfp'$ is the result of regular $\sig$-reducing $\bfp$ w.r.t. $\mathcal{G}$, then $kx^{a}\bfp'$ is regular $\sig$-reduced.
  \end{enumerate}
\end{lemma}
\begin{proof}
  The proof is in two steps: first, we construct a S-polynomial $\bfp$ whose signature divides $\sig(\bfu)$, and then, starting from $\bfp$, we show that there exists an S-polynomial satisfying the conditions of the lemma.

  In the remainder of the proof, for $i \in \{1,\dots,s\}$, let $M(i) = \LM(\bbfalpha_{i})$, $C(i) = \LC(\bbfalpha_{i})$, $T(i) = \LT(\bbfalpha_{i})$ and $S(i) = \sig(\bfalpha_{i})$.

  \vspace{0.3cm}
  \paragraph{Existence of  a S-polynomial satisfying 1.}
  For the first step, let $\sig(\bfu)$ be $lx^b\bfe_i$ for some $l \in \RR$, $x^b$ a monomial 
  and $\bfe_i$ a basis vector.
  Let $\bfe'_{i}$ be the result of regular $\sig$-reducing $\bfe_{i}$.
  If $\bbfe'_{i}=0$, then $\bfu$ regular $\sig$-reduces to $0$, which is a contradiction since we assumed $\bfu$ to be $\sig$-reduced and $\bbfu \neq 0$.
  Let $\bfL = lx^{b}\bfe'_{i}$, it has signature $lx^{b}\bfe_{i}$.
  Then $\bfu - \bfL$ has a smaller signature than $\bfu$, so it $\sig$-reduces to zero and in particular it is $\sig$-reducible. 
  Also, $\bfL$ is $\sig$-reducible by $\bfe'_{i}$.
  Consider the sum $(\bfu-\bfL) + \bfL = \bfu$.
  It is not $\sig$-reducible, which implies that $\LT(\overline{\bfu-\bfL}) = -\LT(\overline{\bfL})$. 

  Let $J_{\LM(\bbfL)}$ be the maximal regular saturated set $J$ with $M(J) \divides \LM(\bbfL)$.
  Since $\bfu - \bfL$ $\sig$-reduces to zero, there exists $(m_{j})_{j \in J_{\LM(\bbfL)}}$ monomials in $\AA$, and $(k_{j})_{j \in J_{\LM(\bbfL)}}$ coefficients in $R$ such that
  \begin{equation}
    \label{eq:21}
    \LT(\bbfu - \bbfL) = \sum_{j \in J_{\LM(\bbfL)}} k_{j}m_{j} T(j)
  \end{equation}
  with $m_{j}M(j) = \LM(\bbfu-\bbfL)$
  and $\sig(k_{j}m_{j}\bfalpha_{j}) = k_{j}m_{j} S(j) \preceq \sig(\bfu - \bfL) \prec \sig(\bfu)$ for all $i$ such that $k_{j} \neq 0$.
  Let $\sigma$ be the index of $\bfe'_i$ in $\mathcal{G}$, that is $\bfalpha_{\sigma} = \bfe'_{i}$.
  Consider the set 
  \(
  J' = \{j : m_{j} \neq 0\} \cup \{\sigma\} \subseteq J_{\LM(\bbfL)},
  \) 
  it is regular by construction.
  
  Let $J$ be a regular saturated set containing $J'$.
  Then, since for all $j \in J'$, $M(j) \divides \LM(\bbfL) = x^{b}M(\sigma)$, 
  \( M(J) =  \lcm\left\{ M(j) : j \in J'\right\} \divides x^{b}M(\sigma). \)
  Furthermore, looking at the leading coefficients in Eq.~\eqref{eq:21}, we have
  \begin{equation}
    \label{eq:23}
    l \, C(\sigma) = - \sum_{j \in J'} k_{j}C(j)
  \end{equation}
  and so
  $l \in \langle C(j)  : j \in J, j \neq \sigma\rangle : \langle C(\sigma) \rangle$.
  Since $\RR$ is a PID, this ideal is principal.
  Let $b_{J}$ be its generator, then $b_{J} \divides l$.
  Let $\bfp$ be the S-polynomial corresponding to $J$ and $b_{J}$.
  It is regular by construction since $J$ is a regular saturated set, and its signature is
  $\sig(\bfp) = b_{J} \frac{M(J)}{M(\sigma)} S(\sigma)
  = b_{J} \frac{M(J)}{\LM(\bbfe'_i)} \bfe_{i}$.
  Since $b_{J}$ divides $l$ and $M(J)$ divides $x^{b}M(\sigma)$,
  $\sig(\bfp)$ divides $lx^{b}s \bfe'_{i} = \sig(\bfL) = \sig(\bfu)$.

  \vspace{0.3cm}
  \paragraph{Existence of a S-polynomial satisfying 1. and 2.}
  Let $\bfp$ be an S-polynomial whose signature divides $\sig(\bfu)$,
  and let $\bfp'$ be the regular $\sig$-reduced form of $\bfp$.
  Write $\sig(\bfu) = \sig(kx^a\bfp)$, where $k \in \RR$ and $x^a$ is a monomial.

  We can assume that $kx^a\bfp'$ is regular $\sig$-reducible or else we are done.
  We then construct an S-polynomial $\bfq$ such that $\sig(lx^b\bfq) = \sig(\bfu)$
  and $\LM(\overline{kx^a\bfp}) \succ \LM(\overline{lx^b\bfq})$.
  If $lx^b\bfq'$, where $\bfq'$ is obtained by regular $\sig$-reducing $\bfq$, is not regular $\sig$-reducible then we are done.
  Otherwise we can do the same process again and get a third S-polynomial with the same properties and keep repeating.
  Since the initial terms are strictly decreasing and we have a well order there are only finitely many such S-polynomials.

  First, we show that we can assume that $x^{a} \succ 1$.
  Indeed, assume that $a=0$ and $k\bfp'$ is regular $\sig$-reducible.
  Since $\RR$ is an integral domain, $\LM(k\bbfp') = \LM(\bbfp')$.
  Let $J_{\LM(\bbfp')}$ be the maximal regular saturated set $J$ with $M(J) \divides \LM(\bbfp')$.
  Then $k\LC(\bbfp')$ lies in the ideal $\langle \LC(\alpha_{j}) : j \in J_{\LM(\bbfp')} \rangle$.
  Since $\RR$ is a PID, this ideal is principal, let $b_{J_{\LM(\bbfp')}}$ be its generator,
  then $b_{J_{\LM(\bbfp')}} \divides k$.
  Let $\bbfq$ be the S-polynomial corresponding to the regular saturated set $J_{\LM(\bbfp')}$
  and the generator $b_{J_{\LM(\bbfp')}}$, its signature divides $\sig(\bfu)$ and is strictly divisible by $\sig(\bfp)$.
  Repeating the process as needed, we obtain a strictly increasing sequence of elements dividing the coefficient of $\sig(\bfu)$, and since $\RR$ is a PID and in particular a unique-factorization domain, this sequence has to be finite.
  So we can assume that $x^{a} \succ 1$.

  We will construct two reductions of $\LT(kx^{a}\bbfp')$, which taken together will give the S-polynomial $\bfq$.
  For the first reduction, the module element $\bfp' \in \AA^{m}$ is regular $\sig$-reduced modulo the \sssGB{} $\mathcal{G}$, and its signature is smaller than $\sig(\bfu)$.
  Furthermore, by assumption $kx^{a}\bfp'$ is not regular $\sig$-reduced, so $\bbfp'$ cannot be $0$.
  So, by definition of a \sssGB{}, $\bfp'$ is 1-singular $\sig$-reducible.
  So there exists
  $(t_{i}^{(1)})_{i \in J_{1}}$ terms in $\AA$, with $J_{1} \subset \{1,\dots,s\}$ and for all $i \in J_{1}$,
  $t_{i}^{(1)} \neq 0$,
  and such that
  \begin{equation}
    \label{eq:equation1}
    \LT(\bbfp') = \sum_{i \in J_{1}} t^{(1)}_{i} \LT(\bbfalpha_{i}) = \sum_{i \in J_{1}} t^{(1)}_{i} T(i)
  \end{equation}
  with for all $i \in J_{1}$, $\LM(t_{i}^{(1)}\bbfalpha_{i}) = \LM(t_{i}^{(1)}) M(i) = \LM(\bbfp')$.
  Furthermore, there exists $\tau$ in $J_{1}$, $t_{\tau}^{(1)} S(\tau) = \sig(\bbfp)$
  and for all $i \in J_{1} \setminus \{\tau\}$, $t_{i}^{(1)} S(i) \prec \sig(\bfp)$.

  We now build the second reduction.
  Since $kx^a\bfp'$ is regular $\sig$-reducible, there exists $(t_{i}^{(2)})_{i \in J_2}$ terms in $\AA$, with
  $J_{2} \subset \{1,\dots,s\}$ and for all $i \in J_{2}$, $t_{i}^{(2)} \neq 0$,
  such that
  \begin{equation}
    \label{eq:equation2}
    \LT(kx^{a}\bbfp') = \sum_{i \in J_{2}} t_{i}^{(2)} \LT(\bbfalpha_{i}) = \sum_{i \in J_{2}} t_{i}^{(2)} T(i),
  \end{equation}
  and for all $j \in J_{2}$, $\LM(t_{j}^{(2)}) M(j) = \LM(kx^{a}\bbfp')$
  and $t_{j}^{(2)}S(j) \prec \sig(kx^a\bfp')$.

  Now let $J = J_{1} \cup J_{2}$, and let $t_{i}^{(1)}=0$ if $i \in J_{2} \setminus J_{1}$, $t_{j}^{(2)} = 0$
  if $j \in J_{1} \setminus J_{2}$.
  Note that $\tau \notin J_{2}$, so $t_{\tau}^{(2)} = 0$.
  Combining Eqs.~\eqref{eq:equation1} and \eqref{eq:equation2}, we obtain a decomposition of $kx^{a}t_{\tau}T(\tau)$ as
  \begin{equation}
    \label{eq:7}
    kx^{a}t_{\tau}T(\tau) = - \sum_{i \in J \setminus \{\tau\}} t_{i}T(i).
  \end{equation}
  where for all $i \in J$, $t_{i} = kx^{a}t_{i}^{(1)} - t_{i}^{(2)}$.
  Furthermore, for all $i \in J \setminus \{\tau\}$,
  $\LM(t_{i})M(i) = \LM(x^{a}\bbfp') = \LM(x^{a} t_{\tau}) M(\tau)$
  and $t_{i}S(i) \prec \sig(\bbfp) = k x^{a} t_{\tau} S(\tau)$.

  The same argument as the one used, in the first part of the proof, to construct an S-polynomial based on Eq.~\eqref{eq:21} yields an S-polynomial $\bfq$ such that $\sig(\bfq)$ divides $\sig(\bfu)$, say $lx^{b}\sig(\bfq) = \sig(\bfu)$.
  Furthermore, since the leading term is eliminated in the construction of an S-polynomial, $\LT(lx^{b}\bbfq) \prec \LT(kx^{a}\bbfp')$, which concludes the proof.
\end{proof}

\begin{theorem}[Correctness of Algorithm~\algosigmoller]
  \label{regularreducers}
  Let $\bfT$ be a term of $\AA^m$
  and let $\mathcal{G}= (\bfalpha_{1},\dots, \bfalpha_{s}) \subseteq \AA^m$ be a finite basis as computed by Algo.~\ref{Algorithmsig}.
  Assume that all regular S-polynomials $\bfp$ with $\sig(\bfp) \prec \bfT$ $\sig$-reduce to $0$ \wrt{} $\mathcal{G}$ .
  Then $\mathcal{G}$ is a semi-strong signature-Gröbner basis up to signature $\bfT$.
\end{theorem}

\begin{proof}
  To get a contradiction assume there exists a $\bfu \in \AA^{m}$ with $\sig(\bfu) \prec \bfT$ such that $\bfu$ does not $\sig$-reduce to zero. 
  Assume w.l.o.g. that $\sig(\bfu)$ is $\prec$-minimal such that $\bfu$ does not $\sig$-reduce to zero and also that $\bfu$ is regular $\sig$-reduced. 

  By Lem.~\ref{supportinglemma} there is an S-polynomial $\bfp$ with $\sig(kx^{a}\bfp) = \sig(\bfu)$ with $k \in \RR$, $x^{a} \in \Mon(\AA)$.
  Also, $kx^a\bfp'$ is regular $\sig$-reduced where $\bfp'$ is the result of regular $\sig$-reducing $\bfp$.

  Let $J_{\LM(\bbfu)}$ be the maximal regular saturated set $J$ with $M(J) \divides \LM(\bbfu)$.
  Since $\sig(kx^a\bfp) = \sig(\bfu)$ and both $kx^a\bfp'$ and $\bfu$ are regular $\sig$-reduced, we have by Lem.~\ref{initialtermslemma} that $\LM(kx^a\bbfp') = \LM(\overline{\bfu} )$, and either $\LT(kx^{a}\bbfp') = \LT(\bbfu)$, or
  \begin{equation}
    \label{eq:12}
    \LC(\bbfu
    -
    kx^{a}\bbfp') \in
    \left\langle \LC(\bbfalpha_{j}) : j \in J_{\LM(\bbfu)} \right\rangle.
  \end{equation}
  So in either case, there exists $(t_{i})_{i \in J_{\LM(\bbfu)}}$ terms in $\AA$, possibly all zero, such that
  \begin{equation}
    \label{eq:18}
    \LT(\bbfu) - \LT(kx^{a}\bbfp') = \sum_{i \in J_{\LM(\bbfu)}} t_{i} \LT(\bbfalpha_{i})
  \end{equation}
  and $t_{i}\LM(\bbfalpha_{i}) = \LM(\bbfr) = \LM(\bbfu)$ for all $i$ such that $t_{i} \neq 0$.
  
  Since $\bfp'$ is a regular S-polynomial with $\sig(\bfp') \preceq \sig(\bfu) \prec \bfT$, $\bfp'$ is $\sig$-reducible, and so $k x^{a} \bfp'$ is $\sig$-reducible.
  So there exists $(\tau_{i})_{i \in J_{\LM(\bbfu)}}$ terms in $\AA$ such that
  \begin{equation}
    \label{eq:19}
    \LT(kx^{a}\bbfp') = \sum_{i \in J_{\LM(\bbfu)}} \tau_{i} \LT(\bbfalpha_{i}),
  \end{equation}
  and $\tau_{i}\LM(\bbfalpha_{i}) = \LM(kx^{a}\bbfp') = \LM(\bbfu)$ for all $i$ such that $\tau_{i} \neq 0$.
  So
  \begin{equation}
    \label{eq:20}
    \LT(\bbfu) = \left( \LT(\bbfu) - \LT(kx^{a}\bbfp') \right)
    + \LT(kx^{a}\bbfp')
    = \sum_{i \in J_{\LM(\bbfu)}} (t_{i}+\tau_{i}) \LT(\bbfalpha_{i}),
  \end{equation}
  and $\bfu$ is $\sig$-reducible which is a contradiction. 
\end{proof}

\subsection{Proof of termination}

The usual proofs of termination of signature-based Gröbner basis algorithms (\eg{}~\cite[Th.~11]{practicalgrobner:2012:stillman}) rely on the fact that all elements which are singular $\sig$-reducible are discarded in the computations.
Algorithm~\algosigmoller only discards those which are 1-singular $\sig$-reducible.
For this reason, we adapt the proof of termination of Algorithm~\alname{RB} \cite[Th.~20]{eder:2013:signature}, which handles singular $\sig$-reducible elements in a different way. 

\begin{theorem}
  \label{lemmatermination}
  Algorithm \algosigmoller terminates.
\end{theorem}
\proofvspace
\proofvspace
\begin{proof}
  Let $\mathcal{G} = (\bfalpha_{1},\dots,\bfalpha_{t},\dots)$ be the sequence of basis elements computed by \algosigmoller.
  By construction, for all $t \geq 1$, $\overline{\bfalpha_t}$ is not $\sig$-reducible by $\mathcal{G}_{t-1} := \{\bfalpha_{1},\dots,\bfalpha_{t-1}\}$, and all $\bfv \in \AA^{m}$ with $\sig(\bfv) \prec \sig(\bfalpha_t)$ $\sig$-reduce to zero w.r.t. $\mathcal{G}_{t-1}$. 

  For $i \geq 1$, let $M(i) = \LM(\bbfalpha_{i})$, $T(i) = \LT(\bbfalpha_{i})$.
  We define the sig-lead ratio $r(\bfalpha_{i})$ of $\bfalpha_{i}$ as $\frac{\sig(\bfalpha_{i})}{M(i)}$.
  Those ratios are ordered naturally by $\frac{s}{m} \prec \frac{s'}{m'} \iff sm' \prec s'm$.
  
  We partition $\mathcal{G}$ into subsets $\mathcal{G}_{r} = \{\bfalpha_{i} \mid r(\bfalpha_{i}) \simeq r\}$, where $\simeq$ denotes equality up to a coefficient in $\RR$.
  We prove that only finitely many $\mathcal{G}_{r}$ are non-empty, and that they are all finite, hence $\mathcal{G}$ is finite.
  
  First, we prove that only finitely many $\mathcal{G}_{r}$ are non-empty.
  We do so by counting minimal basis elements, where $\bfalpha_{i}$ is minimal if and only if there is no $\bfalpha_{j} \in \mathcal{G}$ with $\sig(\bfalpha_{j}) \divides \sig(\bfalpha_{i})$ and $T(j) \divides T(i)$.
  A non-minimal module element $\bfalpha_{i}$ is $\sig$-reducible by $\{\bfalpha_{1},\dots,\bfalpha_{i-1}\}$ (\cite[Lem.~12]{practicalgrobner:2012:stillman}), and since all basis elements are regular $\sig$-reduced by construction, $\bfalpha_{i}$ is singular $\sig$-reducible.
  In particular, there exists at least one $\bfalpha_{j}$, $j < i$
  and a monomial $m$ with $\sig(m\bfalpha_{j}) \simeq \sig(\bfalpha_{i})$ and $m M(j) = M(i)$, so $\bfalpha_{i}$ and $\bfalpha_{j}$ lie in the same subset $\mathcal{G}_{r}$.
  Hence there are at most as many non-empty $\mathcal{G}_{r}$'s as there are minimal basis elements.
  This is finitely many because $\AA$ and $\AA^{m}$ are Noetherian.

  Then we prove by induction on the finitely many non-empty sets $\mathcal{G}_{r}$ that each $\mathcal{G}_{r}$ is finite.
  Let $r$ be a sig-lead ratio, assume that for all $r' < r$, $\mathcal{G}_{r'}$ is finite.
  Let $\bfalpha_{t} \in \mathcal{G}_{r}$.
  If $\bfalpha_{t}$ is $\bfe_{i}$ for some $i$, then it only counts for one.
  Otherwise, let $J$ be the regular saturated set, and $\bfp$ the corresponding S-polynomial, that \algosigmoller regular $\sig$-reduced to obtain $\bfalpha_{t}$.
  Then $\bfp = \sum_{j \in J} b_{j}\frac{M(J)}{M(j)}\bfalpha_{j}$ for $b_{j} \in \RR$, and there exists $\tau \in J$ such that for all $j \in J \setminus \{\tau\}$,
  $\frac{M(J)}{M(j)}\sig(\bfalpha_{j}) \prec \frac{M(J)}{M(\tau)}\sig(\bfalpha_{\tau})$.
  Also $T(t) \prec \LT(\frac{M(J)}{M(\tau)}\bbfalpha_{\tau})$ and
  $\sig(\bfalpha_{t}) = \frac{M(J)}{M(\tau)}\sig(\bfalpha_{\tau})$.
  So \(
  r = \frac{\sig(\bfalpha_{t})}{M(t)}
  \succ \frac{\sig(\bfalpha_{\tau})}{M(\tau)}
  \succ \frac{\sig(\bfalpha_{j})}{M(j)}
  \)
  for $j \in J \setminus \{\tau\}$.
  Hence all $\bfalpha_{j}$, $j \in J$ are in some $\mathcal{G}_{r_{j}}$ with $r_{j} < r$, so for computing elements of $\mathcal{G}_{r}$,
  the algorithm will consider at most as many saturated subsets
  as there are subsets of $\bigcup_{r' < r} \mathcal{G}_{r}$, which is finite by induction.
  Furthermore, since $\RR$ is a PID and in particular Noetherian, with each saturated subset $J$, the algorithm only builds finitely many S-polynomials (actually, it only builds one).
  So overall, we find that $\mathcal{G}_{r}$ is finite, which concludes the proof by induction.
\end{proof}

\subsection{Eliminating S-polynomials}
\label{sec:Elim-S-vect}


It is well known in the case of fields that additional criteria can be implemented to detect that a regular S-pair will lead to an element which $\sig$-reduces to 0.
In this section, we show how we can implement three such criteria, namely the syzygy criterion, the F5 criterion and the singular criterion.

\subsubsection{Syzygy Criterion}

Syzygy criteria rely on the fact that, if the signature of an S-polynomial can be written as a linear combination of signatures of syzygies, then this S-polynomial would be a syzygy itself.
Signatures of syzygies can be identified in two ways:
\begin{itemize}
  \item the Koszul syzygy between basis elements $\bfp$ and $\bfq$ such that $\sig(\bfp)=m_{\bfp} \bfe_{i}$, $\sig(\bfq) = m_{\bfq} \bfe_{j}$, $i < j$ is $\bbfp \bfq - \bbfq \bfp$, and it has signature $\LT(\bbfp)\sig(\bfq)$;
  \item if a regular S-polynomial $\bfp$ $\sig$-reduces to $0$, then $\sig(\bfp)$ and its multiples are signatures of syzygies; thus, the algorithm may maintain a set of generators of signatures of syzygies by adding to this set $\sig(\bfp)$ for each S-polynomial $\bfp$ $\sig$-reducing to 0.
\end{itemize}
For regular sequences, all syzygies are Koszul syzygies.

\begin{proposition}[Syzygy criterion]
  \label{lemme:1}
  Assume that $\bfT$ is a signature such that all module elements with signature less than $\bfT$ $\sig$-reduce to $0$.
  Let $\bfp \in \AA^{m}$ be such that there exist syzygies $\bfz_{1},\dots,\bfz_{k}$
  and terms $m_{1},\dots,m_{k}$ in $\AA$ with $\sig(\bfp) = \sum_{i=1}^{k} m_{i} \sig(\bfz_{i})$, and $\sig(\bfp) \preceq \bfT$.
  Then $\bfp$ regular $\sig$-reduces to $0$.
\end{proposition}
\proofvspace
\begin{proof}
  Let $\bfr = \bfp - \sum_{i=1}^{k} m_{i}\bfz_{i}$, then $\sig(\bfr) \prec \sig(\bfp) \preceq \bfT$ so $\bfr$ $\sig$-reduces to $0$.
  But $\bbfr = \bbfp - \sum_{i=1}^{k} m_{i}\bbfz_{i} = \bbfp$,
  so $\bfp$ also $\sig$-reduces to 0 with reducers of signature at most $\sig(\bfr) \prec \sig(\bfp)$.
\end{proof}

Koszul syzygies can be eliminated with the same technique, but it is more efficient to use the F5 criterion~\cite[Sec.~3.3]{practicalgrobner:2012:stillman}.

\begin{proposition}[F5 criterion, \cite{Faugere:2002:F5,BFS14}]
  \label{lemme:3}
  Let $\bfp \in \AA^{m}$ with signature $\mu\,\bfe_{i}$, and let $\{\bfalpha_{1},\dots,\bfalpha_{t}\}$ be a signature Gröbner basis of $\langle f_{1},\dots,f_{i-1} \rangle$.
  Then $\bfp$ is a Koszul syzygy if and only if $\mu$ is $\sig$-reducible modulo $\{\bfalpha_{1},\dots,\bfalpha_{t}\}$.
\end{proposition}
\proofvspace
\begin{proof}
  By definition, $\bfp$ is a Koszul syzygy if and only if $m \in \LT(\langle f_{1},\dots,f_{i-1} \rangle)$, and the conclusion follows by definition of a weak Gröbner basis.
\end{proof}

\subsubsection{Singular Criterion}

The singular criterion states that the algorithm only needs to consider one S-polynomial with a given signature.
So when computing a new S-polynomial, if there already exists a $\sig$-reduced module element with the same signature, we may discard the current S-polynomial without performing any $\sig$-reduction.

\begin{proposition}[Singular criterion]
  \label{lemme:2}
  Let $\mathcal{G} = \{\bfalpha_{1},\dots,\bfalpha_{s}\}$ be a signature Gröbner basis up to signature $\bfT$.
  Let $\bfp \in \AA^{m}$ be such that there exists $\bfalpha_{i} \in \mathcal{G}$ with $\sig(\bfalpha_{i}) = \sig(\bfp)$ and $\sig(\bfp) = \sig(\bfT)$.
  Then $\bfp$ $\sig$-reduces to 0.
\end{proposition}
\proofvspace
\begin{proof}
  Let $\bfp'$ be the result of regular $\sig$-reducing $\bfp$ \wrt{} $\mathcal{G}$.
  By construction, the basis element $\bfalpha_{i}$ is regular $\sig$-reduced \wrt{} $\mathcal{G}$.
  So by Lem.~\ref{initialtermslemma}, $\LM(\bbfp') = \LM(\bbfalpha_{i})$, and applying Lem.~\ref{lemme:correctness-1sing}, with $k=1$ and $x^{a}=1$, shows that $\bfp'$ is $1$-singular $\sig$-reducible.
  The result of that $\sig$-reduction has signature $\prec \sig(\bfp) = \bfT$, so it $\sig$-reduces to 0.
\end{proof}

\largevspace
\largevspace
\section{Experimental results and future work}
\label{sec:Implementation}

We have written a toy implementation of Algo. \algosigmoller\footnote{\url{https://github.com/ThibautVerron/SignatureMoller}}, with the F5 and Singular criteria.
We provide functions \algocoefs and \algosatideal for Euclidean rings, fields and multivariate polynomial rings.

Since our focus is on the feasibility of signature-compatible computations and not their efficiency, we give data about the number of considered S-polynomials, saturated sets and reductions to 0, when computing Gröbner bases over $\ZZ$ for the polynomial systems Katsura-2~(Table~\ref{tab:bench-katsura2}) and Katsura-3~(Table~\ref{tab:bench-katsura3}).
The statistics are compared with a run of Möller's strong algorithm~\cite{Moller:1988:grobnerrings2}.
Even though the proposed algorithm, adapted from Möller's weak algorithm, considers more saturated sets than Möller's strong algorithm, thanks to the signatures, it ends up computing and reducing significantly less S-polynomials, and no reductions to zero appear.

Running Algo.~\algosigmoller on larger examples would require
optimizations, but it appears that the most expensive step is the
generation of the saturated sets, which takes time exponential in the
size of the current basis.
This step may be accelerated in different ways.
First, it is known that in the case of PIDs, the reductions of
Möller's algorithm can be recovered from those of Möller's strong
algorithm~\cite[Sec.~4.4]{Adams:1994:introtogrobnerbasis}, which may
allow to run the algorithms considering only pairs instead of
arbitrary tuples of polynomials.
Additionally, Gebauer and Möller's criteria for fields can be used to
make Möller's strong algorithm over PIDs more
efficient~\cite{Moller:1988:grobnerrings2}.
We will investigate whether it is possible to prove that these
algorithms are compatible with signatures in the future.
Finally, future research will be focused on further signature-based
criteria, such as the cover criterion described in~\cite{GVW2016} and
the more general rewriting criteria.

The algorithm accepts as input polynomials over any ring, provided that the necessary routines are defined.
In particular, our implementation can run the algorithms on polynomials on the base ring $\KK[y_{1},\dots,y_{k}]$.
On small examples in this setting, it appears that the algorithm terminates and returns a correct output.
Understanding the behavior of \algosigmoller over UFDs or even more general rings will also be the focus of future research.


\renewcommand{\pbox}[2][l]{%
  \begin{tabular}[c]{@{}#1@{}}#2\end{tabular}%
}

\begin{table}
  \centering
  \caption{Computation of a grevlex GB of the Katsura-2 system in $\ZZ[X_1,X_2,X_3]$ }
  \label{tab:bench-katsura2}
  \begin{tabular}{lccc}
    \hline
    Algorithm & Pairs / sat. sets & S-polynomials &  Reductions to 0
    \\
    \hline
    Möller strong & 78 & 20 & 7 \\
    \algosigmoller (with criteria) & 170 & 13 & 0 \\
    \hline
  \end{tabular}
\end{table}

\begin{table}
  \centering
  \caption{Computation of a grevlex GB of the Katsura-3 system in $\ZZ[X_1,X_2,X_3,X_{4}]$ }
  \label{tab:bench-katsura3}
  \begin{tabular}{lccc}
    \hline
    Algorithm & Pairs / sat. sets & S-polynomials & Reductions to 0
    \\
    \hline
    Möller strong & 861 & 246 & 159 \\
    \algosigmoller (with criteria) & 2227 & 51 & 0 \\
    \hline
  \end{tabular}
\end{table}




\largevspace
\smallvspace

\subsection*{Acknowledgments}
The authors thank C.~Eder and anonymous referees for helpful suggestions, M.~Ceria and T.~Mora for a fruitful discussion on the syzygy paradigm for Gröbner basis algorithms, and M.~Kauers for his valuable insights and comments all through the elaboration of this work.

\makeatletter
\enddoc@text
\def\enddoc@text{\relax}
\makeatother

\ifappendix
\newpage
\appendix
\section{Example run of Algorithm~\algosigmoller (Algo.~\ref{Algorithmsig})}
\label{sec:Example-run-Algor}

\input{example.tex}
\fi

\end{document}


%% file: example.tex
\newcommand{\sigindex}[1]{#1^{\ast}}

As an illustration, consider the ring $A = \ZZ[x,y]$ with the lexicographical ordering with $x > y$, and the ideal generated by $f_{1}=3xy +x + y^{2}$ and $f_{2}=x^{2}$.

The algorithm maintains a signature Gröbner basis $G$ and a queue of saturated pairs $\mathcal{P}$.
Both are finite ordered sequences (lists), which we denote with square brackets, \eg{} $G = [\bfg_{1},\bfg_{2},\dots,\bfg_{t}]$.
The elements of $G$ are pairs $(\text{polynomial},\text{signature})$, for which we use the notations $\bfg_{i} = (g_{i},\sig(\bfg_{i}))$.
To simplify the notations, given a basis $G = [\bfg_{1},\dots,\bfg_{t}]$, $m$ a monomial of $A$ and $k \in \NN$, we use the notation
\begin{equation}
  \label{eq:5}
  J_{m}^{(k)} = \Sat(m ; \LM(g_{1}),\dots,\LM(g_{k})) = \left\{ i : i \in \{1..k\} : \LM(g_{i}) \divides m \right\}
\end{equation}
for the saturated sets.

In a saturated set, we will use the notation ${}^{\ast}$ to denote those indices which contribute the maximal signature.
For example, in the saturated set $J = \{1,2,\sigindex{4},5,\sigindex{8}\}$, we would have
\begin{equation}
  \label{eq:16}
  S_{J} \simeq \frac{M(J)}{\LM(g_{4})} \sig(\bfg_{4})
  \simeq \frac{M(J)}{\LM(g_{8})}\sig(\bfg_{8})
  \succneq \frac{M(J)}{\LM(g_{i})}\sig(\bfg_{i}) \text{ for $i \in \{1,2,5\}$}.
\end{equation}
If only one index is marked with a star, the saturated set is regular and that index is the signature index.

The algorithm starts with an empty basis $G = []$ and an empty queue of regular saturated sets
$\mathcal{P} = []$.
We first add $f_{1}$ with signature $\bfe_{1}$ to $G$, that is, we add the element $\bfg_{1} = (3xy +x + y^{2}, \bfe_{1})$ to~$G$.

We observe that no saturated sets can be formed, so $G_{1} = [\bfg_{1}]$ is a weak signature Gröbner basis of $\langle f_{1} \rangle$.

We then introduce $f_{2}$, with signature $\bfe_{2}$.
It cannot be reduced modulo $G$, so we add to the basis the element $\bfg_{2}=(x^{2},\bfe_{2})$.

To form regular saturated sets, we consider all possible least common
multiples of leading monomials of $g_{i}$ involving $g_{2}$.
Here, the set of leading monomials is $\{xy,x^{2}\}$, and the only
non-trivial LCM that we can form is $x^{2}y$.
For each least common multiple $m$, the set
\begin{equation}
  J = J_{m}^{(2)} = \Sat(m; \LM(g_{1}),\LM(g_{2}) ) = \{i \in \{1,2\} : \LM(g_{i}) \text{ divides } m\}
\end{equation}
is a saturated set, with $M(J) = m$.
Here, for $m = x^{2}y$, we get $J_{1} = J_{x^{2}y}^{(2)}= \{1,\sigindex{2}\}$.

We have $M(J_{1}) = x^{2}y = x \LM(g_{1}) = y \LM(g_{2})$.
We multiply the corresponding signatures and we compare: here, $x
\sig(\bfg_{1}) = x\bfe_{1} \precneq y\sig(\bfg_{2}) = y\bfe_{2}$, so the
presignature of $J_{1}$ is $S_{J_{1}} = y\bfe_{2}$, and it is regular with
signature index $2$.

We now compute a S-polynomial associated to $J_{1}$, namely $h_{3} =
3yg_{2} - xg_{1} = -x^{2}-xy^{2}$ with signature $\sig(\bfh_{3}) =
3y\sig(\bfg_{2}) = 3y\bfe_{2}$.
Since $\LM(h_{3}) = - \LM(g_{2})$ and $\sig(\bfh_{3}) \succneq
\sig(\bfg_{2})$, $h_{3}$ is regular $\sig$-reducible modulo $G$, and the
result is $g_{3} = -xy^{2}$.
It still has signature $3y\bfe_{2}$, because we only performed a regular $\sig$-reduction.
We add to $G$ the element $\bfg_{3} = (-xy^{2},3y\bfe_{2})$.

The next few steps are identical, so we give a fast-forward version:
\begin{enumerate}
  \item[4.] Regular saturated set $J_{2} = J_{xy^{2}}^{(3)}= \{1,\sigindex{3}\}$ with $M(J_{2}) =
  xy^{2}$ and $S_{J_{2}} = y\bfe_{2}$
  \\ $\rightarrow$ $\bfg_{4} =
  (xy+y^{3}, 9y\bfe_{2})$;
  \item[5.] Regular saturated set $J_{3} = J_{xy}^{(4)} = \{1,\sigindex{4}\}$ with $M(J_{3}) = xy$
  and $S_{J_{3}} = y\bfe_{2}$
  \\ $\rightarrow$ $\bfg_{5} =
  (-x+3y^{3}-y^{2}, 27y\bfe_{2})$;
  \item[6.] Regular saturated set $J_{4} = J_{xy}^{(5)} =  \{1,4,\sigindex{5}\}$ with $M(J_{4}) = xy$
  and $S_{J_{4}} = y^{2}\bfe_{2}$
  \\ $\rightarrow$ $\bfg_{6} = (3y^{4},
  27y^{2}\bfe_{2})$;
  \item[7.] Regular saturated set $J_{5} = J_{xy^{2}}^{(4)} = \{1,3,\sigindex{4}\}$ with $M(J_{5}) =
  xy^{2}$ and $S_{J_{5}} = y^{2}\bfe_{2}$
  \\ $\rightarrow$ $\bfg_{7} =
  (y^{4}, 9y^{2}\bfe_{2})$.
\end{enumerate}
Both $J_{3}$ and $J_{5}$ were added to $\mathcal{P}$ after Step 4 (construction of $\bfg_{4}$).
But at Step 5, since $S_{J_{3}} \precneq S_{J_{5}}$, we have to consider $J_{3}$ first, and keep $J_{5}$ for later.
After Step 5, $\mathcal{P}$ contains both $J_{4}$ and $J_{5}$, whose
presignatures are incomparable: $S_{J_{4}} \simeq S_{J_{5}}$.
So we could have considered $J_{5}$ before $J_{4}$, the result would
still have been correct.


After introducing $\bfg_{7}$, the basis $G$ has $7$ elements:
\begin{multline}
  \label{eq:2}
  G = [
  (3xy+\dots, \bfe_{1}), (x^{2}, \bfe_{2}), (-xy^{2}, 3y\bfe_{2}), (xy+\dots, 9y\bfe_{2}),\\
  (-x + \dots, 27y\bfe_{2}), (3y^{4},27y^{2}\bfe_{2}), (y^{4},9y^{2}\bfe_{2})
  ].
\end{multline}
The queue $\mathcal{P}$ is not empty at this point, but before continuing, we need to form regular saturated sets using the latest addition $\bfg_{7}$.

We go through the same process as before to form saturated sets:
\begin{enumerate}
  \item List all possible least common multiples of leading monomials involving $\LM(g_{7})$: those are $y^{4}, xy^{4}, x^{2}y^{4}$.
  \item For each of them, compute the corresponding saturated set:
  \begin{itemize}
    \item $y^{4}$ gives $J_{y^4}^{(7)} = \{\sigindex{6},\sigindex{7}\}$ with presignature $y^{2}\bfe_{2}$;
    \item $xy^{4}$ gives $J_{xy^4}^{(7)} = \{1,3,4,5,\sigindex{6},\sigindex{7}\}$, with presignature $xy^{2}\bfe_{2}$;
    \item $x^{2}y^{4}$ gives $J_{x^2y^4}^{(7)} = \{1,2,3,4,5,\sigindex{6},\sigindex{7}\}$ with presignature $x^{2}y^{2}\bfe_{2}$.
  \end{itemize}
\end{enumerate}
None of those 3 saturated sets is regular: there is always a signature
collision between $\bfg_{6}$ and $\bfg_{7}$.
For example, with $J_{x^4}^{(7)}$, $\sig(\bfg_{6}) \simeq \sig(\bfg_{7}) \simeq y^{2}\bfe_{2}$.

So we need to make them regular, which is done by forming new sets
with just one of the colliding signatures.
From $J_{y^4}^{(7)}$, we could form $\{\sigindex{6}\}$ and $\{\sigindex{7}\}$, which are trivial.

From $J_{xy^4}^{(7)}$, we can form the regular saturated sets $\{1,3,4,5,\sigindex{6}\}$ and $\{1,3,4,5,\sigindex{7}\}$.
Since the set $\{1,3,4,5,\sigindex{6}\}$ does not contain $7$, it is already in $\mathcal{P}$.
And we add the new regular saturated set $\{1,3,4,5,\sigindex{7}\}$ to $\mathcal{P}$.

Similarly, from $J_{x^2y^4}^{(7)}$, we find the new regular saturated set $\{1,2,3,4,5,\sigindex{7}\}$ to add to $\mathcal{P}$.

Then we continue with the regular saturated set $J_{6} = \{1,3,4,\sigindex{5}\}$ with $M(J_{2}) = x^{2}y$ and $S(J_{2}) = 27y^{3}\bfe_{2}$.
It gives rise to $h_{8} = 3y^{5} + y^{4}$ with signature $\sig(\bfh_{8}) = 27y^{3}\bfe_{2}$.

Since $\LM(h_{8}) = y\LM(g_{6})$ and $\sig(\bfh_{8}) = y \sig(\bfg_{6})$, we
know that $h_{8}$ is $1$-singular reducible modulo $G$ and can be
discarded.
Note that we only needed to compare the leading \emph{monomials}
(without coefficients) of $h_{8}$ and $g_{6}$, and not verify whether
there is an actual linear combination eliminating that term.

\begin{remark}
  If we had considered
  $J_{5}$ before $J_{4}$ at Step 6, $\bfg_{7}$ would have been built before
  $\bfg_{6}$, and $\bfg_{6}$ would have been discarded for being
  1-singular reducible modulo $\bfg_{7}$.
  In that case, the non-regular saturated sets $J_{y^4}^{(7)}$,
  $J_{xy^4}^{(7)}$ and $J_{x^2y^4}^{(7)}$ would never have been considered.
\end{remark}

The remainder of the run proceeds differently, depending on whether the
F5 criterion (Prop.~5.8) is implemented.
If it is not, the remaining regular saturated sets all give rise to
polynomials regular $\sig$-reducing to $0$, and the algorithm
terminates, returning the $7$-elements basis written above.

On the other hand, if the F5 criterion is implemented, those reductions to zero are excluded.
Let us illustrate it with the next saturated set in the queue: $J_{7}
= \{2,\sigindex{5}\}$, with signature $27xy\bfe_{2}$.
We need to test whether $27xy$ lies in the ideal of leading terms of
$\langle f_{1} \rangle$, which is equivalent to testing whether $27xy$
is reducible modulo the already computed basis $G_{1} =
\{(3xy^{2}+x+y^{2},\bfe_{1})\}$.
Since it is indeed reducible, this regular saturated set indeed satisfies the F5
criterion, and can be discarded without further calculation.
The criterion eliminates all subsequent regular saturated sets in the same way.
